\documentclass[11pt]{article}
\usepackage{fullpage,amsmath, amssymb, epsfig, algorithm, algorithmic, color}
\usepackage{graphicx,psfrag}
\usepackage{hyperref,tikz}
\usetikzlibrary{%
decorations.fractals,%
decorations.shapes,%
decorations.text,%
decorations.pathmorphing,%
decorations.pathreplacing,%
decorations.footprints,%
decorations.markings}

\newcommand{\tm}[1]{\textrm{#1}}
\renewcommand{\P}[1]{\mathbb{P}\left(#1\right)}
\newcommand{\PP}[2]{\mathbb{P}_{#1}\left(#2\right)}
\newcommand{\I}[1]{\mathbb{I}\left(#1\right)}
\newcommand{\E}[1]{\mathbb{E}\left[#1\right]}

\newcommand{\opt}{{\rm{OPT}}}
\newcommand{\alg}{{\rm{ALG}}}

\newcommand{\e}[1]{e^{-#1}}

\newtheorem{theorem}{Theorem}[section]
\newtheorem{lemma}[theorem]{Lemma}

\newtheorem{proposition}[theorem]{Proposition}
\newtheorem{corollary}[theorem]{Corollary}

\newtheorem{remark}[theorem]{Remark}

\newtheorem{observation}[theorem]{Observation}

\newenvironment{proof}{\noindent{\bf Proof}:}{$\hfill \Box$\\}

\def\eopt {\E{\opt}}
\def\ealg {\E{\alg}}
\def\s {y}
\def\S {Y}
\def\t {z}
\def\T {Z}

\begin{document}

\title{\Large Online Stochastic Matching:     Online Actions Based on Offline Statistics}

\author{Vahideh H. Manshadi
    \thanks{ Department of Electrical Engineering, Stanford University, Stanford, CA 94305.
    Email:\protect\url{vahidehh@stanford.edu}.}
\and
Shayan Oveis Gharan%
    \thanks{ Department of Management Science and Engineering, Stanford University, Stanford, CA 94305. Email:\protect\url{{shayan,saberi}@stanford.edu}.}
\and
  Amin Saberi\footnotemark[2]%
}


\maketitle

\pagenumbering{arabic}
\setcounter{page}{1}

\begin{abstract}
We consider the online stochastic matching problem proposed by Feldman et al. \cite{aryanak_stmatching} as a
model of display ad allocation. We are given a bipartite graph; one side of the graph
corresponds to a fixed set of bins and the other side represents the set of possible
ball types. At each time step, a ball is sampled
independently from the given distribution and it needs to be matched upon its arrival to an empty bin.
The goal is to maximize the number of allocations.

We present an online algorithm for this problem with a competitive ratio of $0.702$. Before our result, algorithms with a competitive ratio better than $1-1/e$ were  known under the assumption that the expected number of arriving balls of each type is integral. A key idea of the algorithm
is to  collect statistics about the decisions of the optimum offline solution using Monte Carlo sampling and use those statistics  to guide the decisions of the online algorithm.  We also show that
our algorithm achieves a competitive ratio of $0.705$ when the rates are integral.

On the hardness side, we prove that no online algorithm can have a  competitive ratio better than $0.823$
under the known distribution model (and henceforth under the permutation model). This improves upon the $\frac56$ hardness result proved by Goel and Mehta \cite{aryanak_randominput} for the permutation model.
\end{abstract}


\section{Introduction}
\label{sec:introduction}

We study a natural variation of bipartite matching problem motivated in the context of online advertising: suppose we are given a bipartite graph $G(\S, \T, E)$ where $\S$ is the set of stochastic nodes (or ball types) and $\T$ is the set of non-stochastic nodes (or bins).
At times $t = 1, 2, \cdots b$, a ball of type $\s \in  \S$ is chosen independently at random from a given distribution. The algorithm can assign the ball to at most one of the empty bins that are adjacent to it. Further, each bin can be matched to at most one ball. The goal is to maximize the expected number of non-empty bins at time $b$. We refer to this model as the {\em known distribution model}.

When the balls are chosen by an adversary instead of a random process, Karp, Vazirani, and Vazirani \cite{kvv} gave a simple and elegant randomized algorithm that achieves a competitive ratio of $1-1/e$. We present the first algorithm for this problem that improves the $1-1/e$ competitive ratio for the stochastic version in its general form. Previously, Feldman et al. \cite{aryanak_stmatching} (and later \cite{bahmani}) used a very interesting combinatorial algorithm to show that this is possible when the arrival rate of every ball, that is the expected number of times it appears, is integral (this is also known as the {\em i.i.d. model}). This assumption, even though not very restrictive for the display ad allocation, is somewhat unnatural. For example, when the distribution is uniform, it requires  $b/|\S|$ to be an integer.

One of the key ideas in designing our algorithm is to approximately compute the expected matching used by the optimum offline algorithm and use it  to guide the decisions of the online algorithm. In particular, using Monte Carlo sampling,  one can compute $f_{(\s, \t)}$, the probability that the optimum offline algorithm allocates a ball of type $\s$ to a bin of type $\t$, for every $\s$ and $\t$. Without loss of generality, we can assume $f$ is a fractional matching.

Our first algorithm writes $f$ as a distribution over integral matchings and samples two matchings $M_1$ and $M_2$ from it. Then, in the online phase, it will use these two matchings for allocating the arriving balls to the bins (see Section \ref{sec:nonadaptive}). The analysis of our algorithm is much shorter and simpler than both \cite{aryanak_stmatching,bahmani}. All these algorithms are non-adaptive, in the sense that they decide the allocation of all the balls regardless of the allocation of the bins before they arrive. We present a simple example to show that no non-adaptive algorithm can achieve a competitive ratio better than $1-1/e$ when the arrival rates are non-integral (see Proposition \ref{prop:hardness_fractionalrate}).

The main result of the paper is an {\em adaptive algorithm} that obtains a competitive ratio of $0.702$ for arbitrary rates, and $0.705$ for the i.i.d. model (see Section \ref{sec:adaptive}).
Unlike the non-adaptive algorithms, our adaptive algorithm decides the allocation of each arriving ball based on the current allocation of the bins. In particular, when a ball arrives the algorithm samples
two neighbor bins from a joint distribution and tries to match it to the first bin; if the bin is already matched the algorithm tries the second bin.
To the best of our knowledge, this is the first algorithm that beats the $1-1/e$ ratio in the general form.
The adaptivity of the algorithm imposes a lot of dependencies in the distribution of full bins and because of that our analysis is somewhat intricate.

On the hardness side, we  present an example that gives an upper bound of $0.823$ on the competitive ratio of any deterministic or randomized online algorithm in the known distribution model (see Proposition \ref{prop:best_example}). For analyzing this example, we use the expected size of a maximum matching of a random bipartite graph recently computed by \cite{dgmmpr09,fm09,fk09} in the context of random SAT and cuckoo hashing.

\subsection{Related Work}
\label{subsec:related}

%
%
%


Bipartite matching problems are central in algorithms and combinatorial optimization and arise naturally in several applications such as resource allocation, scheduling, and online advertising.

The online matching problem was first studied by Karp, Vazirani, and Vazirani \cite{kvv} in the adversarial model
where the graph is unknown; when a ball arrives it reveals its incident edges. They proved that a simple randomized
on-line algorithm achieves $(1 - 1/e)$ and this factor is the best possible performance.

More recently, Feldman et al. \cite{aryanak_stmatching} studied the problem under stochastic assumptions.
They assumed that the graph is known but the
sequence of arrivals are $i.i.d.$ samples from a given distribution. Further, they assumed that sampling rates are integral and developed an online
algorithm that beats $(1-1/e)$. They also showed that there is no $1 - o(1)$-approximation algorithm for this
setting. Recently, Bahmani and Kapralov \cite{bahmani} improved the upper and lower bounds of Feldman et al.
to 0.902 and 0.699 respectively  in the same setting. Also, they showed that for $d$-regular graphs, a simple
randomized algorithm achieves a competitive ratio of
$1 - O(1/\sqrt{d})$ \cite{bahmani}.

Goel and Mehta \cite{aryanak_randominput} considered a different stochastic model: they assumed the graph is unknown but the sequence of arrivals is a random permutation. This is known as the {\em random permutation model}, and it is a generalization of the known distribution model. They
showed that a greedy algorithm achieves $(1-1/e)$ factor. Further, they showed that no online algorithm
can achieve competitive ratio better than $\frac{5}{6}$. Since the known distribution model
is a special case of the random permutation model, our hardness result improves their upper-bound to $0.823$. Since the first appearance of this paper, Karande et al. \cite{KMT11}, and Mahdian and Yan \cite{MY11} independently improve the $(1-1/e)$ competitive ratio in the random permutation model to $0.653$ and $0.696$ respectively.

A close line of work to the online matching is the online $b$-matching and the AdWords problem \cite{saberi,hayes_b_matching}. Mehta et al. \cite{saberi} developed a $(1-1/e)$
online algorithm in the adversarial case. Recently, Devanur and Hayes \cite{hayes_b_matching} improved the competitive ratio to $(1-\epsilon)$
in the stochastic case  where the sequence of arrivals is a random permutation or it consists of $i.i.d.$ samples.

\section{Problem Definition}
\label{sec:problemdefinition}

Let $G(\S, \T, E)$ be a bipartite graph where $\S$ is the set of stochastic nodes (or ball types) and $\T$ is the set of non-stochastic nodes (or bins). There is a rate $r_{\s}$ associated to every type of ball $\s \in \S$. The online stochastic matching problem is as follows: at times $t = 1, 2, \cdots b$, a ball of type $\s \in  \S$ is chosen independently and with probability proportional to $r_{\s}$. The algorithm can assign this ball to at most one of the empty bins that are adjacent to it; each bin can be matched to at most one ball. The goal of the algorithm is to maximize the expected number of non-empty bins at time $b$.

Without loss of generality,  we assume that $\sum_{\s\in \S} r_{\s} = b$, thus the expected number of balls of type $\s$ in the sequence is $r_{\s}$. Also, we assume that $r_{\s}\leq 1$; if a node has a rate greater than 1, we can easily split it into a set of identical nodes with rates at most 1.

We will study two classes of algorithms: non-adaptive and adaptive. A non-adaptive algorithm is equivalent to an ordering of the neighbors $N(\s)$ of every node $\s \in \S$. If $\t_1, \t_2, \cdots \t_{|N(\s)|}$ is such an ordering for $\s$, then the $k$-th time a ball of type $\s$ arrives, the algorithm will allocate it to bin $\t_k$  if it is empty. If $k > |N(\s)|$ or $\t_k$ is full then the ball will not be allocated.
 On the other hand, adaptive algorithms can choose the assignment of every ball when it arrives.

We will compare our algorithms to the optimum offline solution. Given the sequence of arrived balls $\omega=(\s_1, \s_2, \ldots,\s_b)$, one can compute the optimum allocation, $\opt(\omega)$, in polynomial time by solving a maximum matching problem. Fix a particular maximum matching algorithm and let $F(\omega): E \rightarrow \{0,1\}$ be the vector indicating which edges are used in the optimum allocation given $\omega$. Clearly, $\opt(\omega)=1^TF(\omega)$ and the competitive ratio of an online algorithm $\alg$ is defined as $\frac{\E{\alg}}{\E{\opt}}$. In our case, both $\alg$ and $\opt$ are concentrated around their expected values, therefore the above competitive ratio is fairly robust (see Feldman et al. \cite{aryanak_stmatching} for a more detailed discussion).

Our algorithms will crucially use the optimum offline solution for making decisions. In particular,
define
\begin{equation} \label{eq:f}
f = \sum_{\omega} F(\omega) \P{\omega},
\end{equation}
where $\P{\omega}$ is the probability of the sequence $\omega=(\s_1, \s_2, \ldots,\s_b)$.
By definition, $f$ is a convex combination of matchings and therefore it is in the convex hull of the matchings of $G$. We will refer to $f$ as the  {\em fractional matching} defined by $\opt$. For each edge $e=(\s,\t) \in E$, $f_e$ is the probability that a ball of type $\s$ is allocated to bin $\t$ by the optimum offline algorithm. Similarly we define the fractional degree of a node to be  $f_{v} = \sum_{e \sim v}f_e$  for $v \in \S \cup \T$.

\begin{proposition}\label{prop:OPT}
The vector $f$ is a fractional matching in $G$. i.e.
\begin{eqnarray}
\label{eq:fmatching_st}
f_{\s}  \leq  r_{\s} \leq 1,~\s \in \S, ~~\tm{and} ~~ f_{\t}  \leq  1,~\t \in \T.
\end{eqnarray}
Moreover, for $e=(\s,\t)$, we have $f_e\leq 1-\e{r_{\s}} + o(1/b).$
\end{proposition}

\begin{proof}
Given $\omega$, let $N_{\s}(\omega)$ be the number of balls of type $\s$ in $\omega$. Clearly $\sum_{e \sim \s}F_e(\omega) \leq N_{\s}(\omega)$. Taking expectations from both sides results in the first inequality in \eqref{eq:fmatching_st}. Similarly, the second inequality in \eqref{eq:fmatching_st} can be proved by noting that in any instance of the problem, $\t$ can be matched to at most one ball.
Finally, for $e = (\s,\t)$, we have
\begin{eqnarray*}
f_e  \leq  \P{N_{\s}(\omega) \geq 1} =  1 - (1-\frac{r_{\s}}{b})^b \leq  1-e^{-r_{\s}} + o(1/b).
 \end{eqnarray*}
\end{proof}

Throughout the paper, we will assume that $b$ is sufficiently large so that $o(1/b)$ is negligible.
We will need to compute $f_e$ for every edge $e$. Obviously, $f_e$'s can be computed by enumeration in time $O(|\S|^b)$.  It is also easy to see that $\eopt$ and $f(e)$ for all $e \in E$, can be approximated with great accuracy using Monte Carlo method.  \opt~is an integral random variable which is in interval $[0,b]$, hence its variance is upper-bounded by $b^2$. Therefore, $\E{\opt}$ can be estimated with error of $o(1/b)$, by averaging over $O(b^3)$ independent samples of the process. A similar argument shows that with $O(|E|^2 b^4)$ samples of
 $\omega$ in equation (\ref{eq:f}), with high probability, one can compute the vector $f$ with accuracy within $o(1/b|E|)$. In the rest of the paper, for simplicity of notation, we will assume that we have estimated $f$ accurately and ignore $o(\cdot)$ terms.

 Since $f$ is a fractional matching, standard algorithmic versions of Caratheodory's theorem (see e.g. \cite[Theorem 6.5.11]{GLS88}) say
that, in polynomial time, we can decompose a feasible solution in the bipartite matching polytope
into a convex combination of polynomially many bipartite matchings.
More specifically, we obtain the following:

\begin{corollary}
\label{cor:marginal_dist}
It is possible to efficiently and explicitly construct (and sample from) a distribution $\mu$ on the set of matchings in $G$ such that
\[
\sum_{M, ~ e\in M} \mu(M) = f_e, ~ \forall e \in E
\]
\end{corollary}

%
%
%
%

\section{A Non-adaptive algorithm}
\label{sec:nonadaptive}
In this section, we will analyze a simple non-adaptive algorithm for the special case where all rates are one, i.e., $r_{\s} = 1, \forall \s \in \S$.
This is the setting studied in Feldman et al. \cite{aryanak_stmatching}. Our algorithm and its analysis is simpler and more intuitive than \cite{aryanak_stmatching}. It also gives a slightly better competitive ratio.

Our non-adaptive algorithm has some similarities with the online algorithm that Feldman et al. propose \cite{aryanak_stmatching}. Both algorithms start by computing two matchings $M_1$ and $M_2$ offline; we use the first matching, only for the first arrived ball of each type and the second one only for the second arrivals. In particular, when the first ball of type $\s$ arrives it will be allocated to the bin matched to $\s$ in $M_1$, and when the second ball arrives, we will allocate it via $M_2$. If the corresponding bins are already full, the balls will be dropped. Note that the probability that there are more than two balls of each type $\s$ in the sequence of arrivals is very small.

On the other hand, we use a different method from \cite{aryanak_stmatching} to construct these matchings. Feldman et al.
find $M_1$ and $M_2$ by decomposing
the solution of a maximum 2-flow of  $G$ into two disjoint matchings (since all the rates are one, the
expected graph is simply $G$). However, we will sample our matchings from the distribution $\mu$ defined by the optimum solution $f$.

\begin{algorithm}
\caption{The Online Non-adaptive Algorithm}
\label{alg:samplematching}
\begin{algorithmic}[1]
\item[{\bf Offline Phase:}]
\STATE Compute the fractional matching $f$, and the distribution $\mu$ using Corollary \ref{cor:marginal_dist}.
\STATE Sample two matchings $M_1$ and $M_2$ from $\mu$ independently; set $M_1$ ($M_2$) to be the first (second)
priority matching.
\item[{\bf Online Phase:}]
\STATE When the first ball of type $\s$ arrives, allocate it through the first priority matching, $M_1$.
\STATE Similarly, when a ball of type $\s$ arrives for the second time, allocate it through the second priority matching, $M_2$.
\end{algorithmic}
\end{algorithm}

The outline of the algorithm is presented in Algorithm \ref{alg:samplematching}.
In the rest of this section, we analyze Algorithm \ref{alg:samplematching}, and show that its approximation ratio is 0.684. Let $X_{\t}$ be the random variable indicating the event that bin $\t$ is matched with a ball during the run of the algorithm. We analyze the competitive ratio of the algorithm by comparing $\E{X_{\t}}$ with $f_{\t}$:
\begin{eqnarray*}
\frac{\ealg}{\eopt} =  \frac{\sum_{\t\in \T} \E{X_{\t}}}{\sum_{\t\in \T} f_{\t}} \geq \min_{\t\in \T} \frac{\E{X_{\t}}}{f_{\t}}
\end{eqnarray*}

Consider any  $\t\in \T$, and  with a slight abuse of notation let $M_1(\t)$ denote the stochastic node matched to it in $M_1$. More precisely, if $(\s, \t) \in M_1$, define $M_1(\t) = \{\s\}$, and if $\t$ is not saturated in $M_1$, define $M_1(\t) = \emptyset$; similarly define $M_2(\t)$. Note
that $\t$ is saturated by $M_1$ (or $M_2$) with probability $f_{\t}$, but if $M_1(\t) = M_2(\t)$, bin $\t$
will only be used for the first arrived ball and effectively it is not saturated by $M_2$.
Given $M_1$ and $M_2$, $\E{X_{\t} | M_1, M_2}$ can be computed similar to \cite[section 4.2.2]{aryanak_stmatching} by considering the following cases:

\begin{eqnarray}
\E{X_{\t} |M_1, M_2} =
\label{eq:nonadapt_allocation}
\begin{cases}
0 & \tm{ if $M_1(\t)=M_2(\t)=\emptyset$} \\
1-1/e & \tm{ if $M_1(\t)\neq \emptyset, \{ M_1(\t)=M_2(\t)\}$} \\
1-1/e & \tm{ if $M_1(\t)\neq \emptyset, M_2(\t)=\emptyset $} \\
1-2/e & \tm{ if $M_1(\t)= \emptyset, M_2(\t)\neq \emptyset$} \\
1-2/e^2 & \tm{ if $M_1(\t)\neq \emptyset, M_2(\t)\neq \emptyset, M_1(\t)\neq M_2(\t)$}
\end{cases}
\end{eqnarray}

By substituting \eqref{eq:nonadapt_allocation} into $\E{X_{\t}}$ we get:
\begin{eqnarray*}
\E{X_{\t}}& = & (1-1/e) \sum_{e\sim \t} f_e (1-f_{\t} + f_e) + (1-2/e) \sum_{e\sim \t}  f_e(1-f_{\t}) 
+ (1-2/e^2) \sum_{e,e'\sim \t,~ e\neq e'} f_e f_{e'} \\
&= & ~f_{\t}(2-3/e)  - f_{\t}^2 (1+2/e^2-3/e)  - (1/e-2/e^2) \sum_{e\sim \t} f_e^2
\end{eqnarray*}

The last equality can be derived by algebraic manipulation and noting that $\sum_{e\sim \t} f_e=f_{\t}$.
It remains to prove a lower bound on the value of the last equation:

\begin{lemma}
\label{lem:NA_optimize}
In any graph $G=(\S,\T,E)$, if $f$ is the corresponding vector of the optimum solution, we have
\begin{equation}
\label{eq:NA_optimize}
\frac{\E{X_{\t}}}{f_{\t}} =  (2-3/e) - (1+2/e^2-3/e) f_{\t} - (1/e-2/e^2) \frac{\sum_{e\sim \t} f_e^2}{f_{\t}} 
 \geq  0.684
\end{equation}
\end{lemma}
\begin{proof}
The proof of this lemma is mainly algebraic. 
Let us first fix $f_{\t}$ and find the minimum of the LHS in terms of $f_{\t}$. For any $f_{\t}$, the LHS is minimized when $\sum_{e\sim \t} f_e^2$ is maximized.
Note that $\sum_{e\sim \t} f_e=f_{\t}$, and thus to maximize the $\sum_{e\sim \t} f_e^2$, we need to consider the most ``unbalanced'' edge probabilities that are
consistent with the properties of fractional matching $f$. By proposition \ref{prop:OPT}, $f_e \leq 1 - \e{1}$ for each $e\sim \t$, thus for $f_{\t} \leq 1 - \e{1}$,
the term $\sum_{e\sim \t} f_e^2$ is maximized when we have only one edge with nonzero probability.
Similarly we can show that the summation of the probabilities of any 2 edges incident to $\t$ is at most $1-\e{2}$, thus if $1 - \e{1} \leq f_{\t} \leq 1 - \e{2}$,
the term $\sum_{e\sim \t} f_e^2$ is maximized when we have two edges with nonzero probability; one edge with probability $1 - \e{1}$  and one with probability
$f_{\t} - (1 - \e{1})$. Similarly we can proceed to compute the maximum of $\sum_{e\sim \t} f_e^2$ in terms of $f_{\t}$ for all $0 \leq f_{\t} \leq 1$.

The only remaining task is to find the value $f_{\t}$ that minimizes the LHS of \eqref{eq:NA_optimize}.
Intuitively, the LHS is minimized when $f_z=1$. In particular, if $f_{\t} < 1$, we may add a dummy node $\s$ to $\S$, and connect it to $\t$ by an edge $e=(\s,\t)$ with very small probability, i.e. $f_e=\epsilon$. It is easy to see that this can only decrease the LHS.
Also, one can numerically confirm that the LHS of \eqref{eq:NA_optimize}
attains its minimum at $f_{\t} = 1$ with value $0.684$.

%
\end{proof}
%
\begin{theorem}
\label{thm:nonadaptive}
Assuming all the rates are 1, the solution of Algorithm \ref{alg:samplematching} is within 0.684 of the optimum offline solution.
\end{theorem}



\section{The Adaptive Algorithm}
\label{sec:adaptive}

In the analysis of the non-adaptive algorithm presented in the previous section, we assumed that the arrival rates of all stochastic nodes are integral and in particular, they are at least one. This is a crucial assumption. If the rates $r_{\s}$'s are not bounded from below, the probability of receiving a second ball of the same type can become arbitrary low and the competitive ratio of the algorithm can get very close to $1-1/e$. This is the case for all non-adaptive algorithms: In Proposition \ref{prop:hardness_fractionalrate} we show that no non-adaptive (even randomized)
algorithm can achieve a competitive ratio better than $1-1/e$ when the sampling rates are not necessarily integral.

%

\newcommand{\fp}[2]{z_{1,#1}(#2)}
\renewcommand{\sp}[2]{z_{2,#1}(#2)}

In this section, we will analyze a simple {\em adaptive} algorithm that will have a better performance for arbitrary rates. The algorithm is very simple: when a ball of type $y$ arrives, it samples two neighboring bins $z_1$ and $z_2$ from a joint distribution. If $z_1$ is empty then $y$ is matched to $z_1$. Otherwise, the algorithm will try $z_2$ and match $y$ to it if it is empty.

The joint distribution from which $z_1$ and $z_2$ are chosen, is determined in advance for every ball type $y$ and it has the following properties: (i) The probability that $z_1$ is equal to $z$ is equal to $f_{(y,z)}$. The same is true for $z_2$. Recall that rates are normalized such that $\sum_{\s\in\S} r(\s)=b$ and thus $f$ is a fractional matching. (ii) The joint distribution is such that the probability of $z_1 = z_2$ is minimized. Note that such a joint probability maximizes the possibility that a ball tries a second different bin in case the first bin that it tries is full.
In what follows, we will present one joint distribution with these properties.

Suppose $(\s,\t_1),\ldots,(\s,\t_k)$ are the edges incident to $\s$, and without loss of generality assume that $f_{(\s,\t_1)} \geq \ldots \geq f_{(\s,\t_k)}$. Also define a dummy edge $(\s,\t_{k+1})$ that is connected to a dummy non-stochastic node $\t_{k+1}$, with
$f_{(\s,\t_{k+1})} = r_{\s} - f_{\s}$. Note that $f_{(\s,\t_{k+1})}$ is the probability that $\opt$ drops a ball of type $\s$.
We will construct two different partitions of the interval $I_{\s} = [0,r_{\s}]$. Specifically, partitions $\mathcal{I}_{\s}$ and $\mathcal{J}_{\s}$ are defined as follows:

\begin{itemize}
\item Partition $\mathcal{I}_{\s}$: let $I_{(\s,\t_1)} = [0,f_{(\s,\t_1)}]$; similarly let $I_{(\s,\t_{l})} = [\sum_{j=1}^{l-1} f_{(y,\t_j)}, \sum_{j=1}^{l} f_{(y,\t_j)}]$, $ 2 \leq l \leq k+1$.
\item Partition $\mathcal{J}_{\s}$: let $J_{(\s,\t_1)}= [r_{\s}-f_{(\s,\t_1)},r_{\s}]$, $J_{(\s,\t_2)} = [0, f_{(\s,\t_2)}]$, and similarly
$J_{(\s,\t_l)} = [\sum_{j=2}^{l-1}f_{(\s,\t_j)}, \sum_{j=2}^{l} f_{(\s,\t_j)}]$, $ 3 \leq l \leq k+1$.
\end{itemize}

Note that the second partition is obtained by shifting the subintervals of $\mathcal{I}_{\s}$ to the
left by $f_{(\s,\t_1)}$. Figure \ref{fig:shift_exmaple} illustrates the partitions through a simple example. Having $\mathcal{I}_{\s}$ and $\mathcal{J}_{\s}$, the distribution is defined as follows: choose a number $x$ uniformly at random from $[0, r_y ]$, define $\fp{\s}{x}$ to be $\t$ if $x \in I_{(\s,\t)}$; similarly define $\sp{\s}{x}$ to be $\t'$ if $x \in J_{(\s,\t')}$. It is easy to see that this joint distribution has property (i). Also, note that the second partition $\mathcal{J}_{\s}$ has the minimum possible overlap with the first one which implies that the resulting joint probability has property (ii), i.e., for each stochastic node $y$, the probability that $\fp{y}{.} = \sp{y}{.}$ is minimized. Further, if all $f_{(y,z)}$'s are less than $\frac{1}{2} r_y$, the probability of $\fp{y}{.} = \sp{y}{.}$ is equal to zero.

\begin{observation}
\label{obs:priority2_slots}
For stochastic node $\s$, suppose $(y,z^*)$ is the edge with the maximum probability,
i.e. $f_{(y,z^*)} \geq f_{(y,z)}$, $\forall z \sim \s$. If $f_{(y,z^*)} < \frac{1}{2}r_{\s}$
then $\fp{y}{x} \neq \sp{y}{x}$, for all $x \in [0,r_{\s}]$. Otherwise, $\fp{y}{x} \neq \sp{y}{x}$
only for $x \in [r_{\s}-f_{(y,z^*)}, f_{(y,z^*)}]$.
\end{observation}

\begin{figure*}
\centering
    \psfrag{r_s}{\textcolor{blue}{\scriptsize{$r_{\s} = 1$}}}
    \psfrag{f_e1}{\textcolor{black}{\scriptsize{$f_{e_1} = 0.5$}}}
    \psfrag{f_e2}{\textcolor{green}{\scriptsize{$f_{e_2} = 0.2$}}}
    \psfrag{f_e3}{\textcolor{yellow}{\scriptsize{$f_{e_3} = 0.2$}}}
    \psfrag{t_1}{\textcolor{black}{\scriptsize{$\t_1$}}}
    \psfrag{t_2}{\textcolor{green}{\scriptsize{$\t_2$}}}
    \psfrag{t_3}{\textcolor{yellow}{\scriptsize{$\t_3$}}}
    \psfrag{I1}{\textcolor{black}{\scriptsize{$I_{e_1}$}}}
    \psfrag{I2}{\textcolor{green}{\scriptsize{$I_{e_2}$}}}
    \psfrag{I3}{\textcolor{yellow}{\scriptsize{$I_{e_3}$}}}
    \psfrag{I4}{\textcolor{red}{\scriptsize{$I_{e_4}$}}}
    \psfrag{Ip1}{\textcolor{black}{\scriptsize{$J_{e_1}$}}}
    \psfrag{Ip2}{\textcolor{green}{\scriptsize{$J_{e_2}$}}}
    \psfrag{Ip3}{\textcolor{yellow}{\scriptsize{$J_{e_3}$}}}
    \psfrag{Ip4}{\textcolor{red}{\scriptsize{$J_{e_4}$}}}
	\includegraphics[height=.2\textheight,clip]{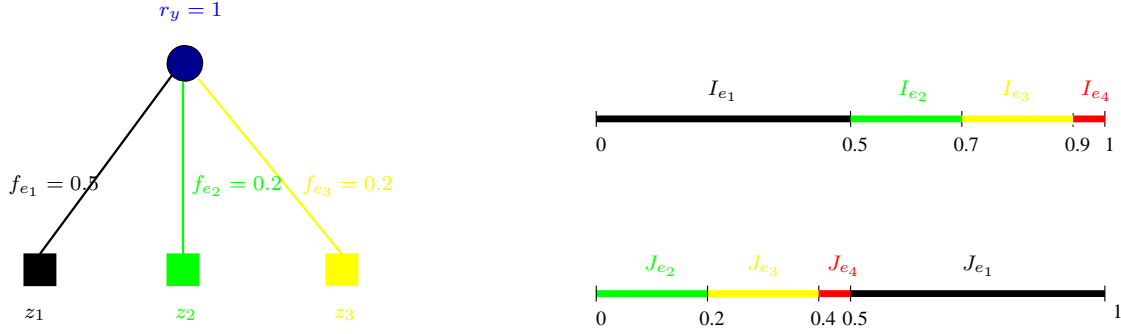}
\caption{Illustration of partitions $\mathcal{I}_{\s}$ and $\mathcal{J}_{\s}$ for node $\s$ with edges $e_1$, $e_2$, and $e_3$}
\label{fig:shift_exmaple}
\end{figure*}

%
%
%
%
%
%
%
%

{\noindent The outline of the algorithm is presented in Algorithm \ref{alg:adapt}.}

\begin{algorithm}
\caption{Online Adaptive Algorithm}
\label{alg:adapt}
\begin{algorithmic}[1]
\item[{\bf Offline Phase:}]
\STATE Compute the fractional matching $f$.
\STATE For each $\s \in \S$ and $x \in [0,r_{\s}]$, construct the functions $\fp{y}{\cdot}$ and $\sp{y}{\cdot}$ by defining the corresponding partitions $\mathcal{I}_{\s}$ and $\mathcal{J}_{\s}$.
\item[{\bf Online Phase:}]
\STATE If a ball of type $\s \in \S$ arrives, choose a number $x$ uniformly at random from interval $[0,r_{\s}]$.
\STATE Match the ball with $\fp{\s}{x}$;
\STATE If $\fp{\s}{x}$ is full, match the ball with $\sp{\s}{x}$;
\end{algorithmic}
\end{algorithm}

\begin{theorem}
\label{thm:adap}
For any graph $G$ and arbitrary set of rates $\{r_{\s},~ \s \in \S\}$, the competitive ratio of Algorithm \ref{alg:adapt} is at least $0.702$.
\end{theorem}

Unlike Algorithm 1, the analysis of Algorithm 2 is fairly intricate, mainly because the adaptivity of the algorithm introduces new dependencies. We will present the proof in a few steps to build an intuition before getting to the actual calculations.

\begin{proof}
Consider a non-stochastic node $\t \in \T$. Bin $\t$ can be matched
as a first priority bin or as a second priority bin. Note that
a bin will be matched once it is tried as a first or second priority.
We define the event $\mathcal{A}_{\t}(t)$ to be the event that bin $\t$ was tried as a first
priority bin by time $t$, i.e., at any time $1, 2, \ldots, t$. Also, define $\mathcal{B}_{\t}(t)$
to be the event that bin $\t$ was tried as a second priority bin at time $t$.
Using the notation defined in the previous section:

\begin{align} \label{eq:Node_Exp}
\E{X_{\t}} =  \P{\mathcal{A}_{\t}(b) \vee \cup_{t=1}^{b} \mathcal{B}_{\t}(t)}
& =  \P{\mathcal{A}_{\t}(b)} + \P{ \cup_{t=1}^{b} \mathcal{B}_{\t}(t)\wedge \overline{\mathcal{A}_{\t}}(b)} \nonumber \\
& =  \P{\mathcal{A}_{\t}(b)} + \P{ \cup_{t=1}^{b} \mathcal{B}_{\t}(t) \Big| ~\overline{\mathcal{A}_{\t}}(b)}\P{\overline{\mathcal{A}_{\t}}(b)}
\end{align}

We need to compute $\P{\mathcal{A}_{\t}(b)}$. Instead we compute $\P{\mathcal{A}_{\t}(t)}$ for $1 \leq t \leq b$;
at each time step, the probability that a ball tries  $\t$ as a first priority bin is equal to the probability that
a ball of type $\s$ arrives, where $\s$ is connected to $\t$ through edge $(y,z)$, and
we choose a point in the interval $I_{(y,z)}$.
This probability is $\frac{\sum_{y \sim \t} f_{(y,z)}}{\sum_{\s \in \S} r_\s} = \frac{f_{\t}}{b}$, and we have:

\begin{eqnarray} \label{eq:priority_one}
\P{\mathcal{A}_{\t}(t)} = 1 - (1-\frac{f_{\t}}{b})^t = 1 - \e{\frac{tf_{\t}}{b}} + o(1/b).
\end{eqnarray}

Thus $\P{\mathcal{A}_{\t}(b)} = 1 - \e{f_{\t}}$. The more difficult part of the analysis is to lower-bound $\P{\cup_{t=1}^{b} \mathcal{B}_{\t}(t)| ~\overline{\mathcal{A}_{\t}}(b)}$. To analyze this probability, we define the parameter $q_{\t}:= \sum_{y \sim \t} \int_{x \in J_{(y,z)} \setminus I_{(y,z)}} 1 dx$. Roughly speaking, we can interpret $q_z$ as the fractional degree of $z$ in the second priority. Note that $q_z \leq f_z$ and the equality holds iff for all $y \sim z$, $I_{(y,z)} \cap J_{(y,z) } = \emptyset$. In Lemma \ref{lem:q_z} we lower-bound $q_z$ in terms of $f_z$. The following lemma lower-bounds $\P{\cup_{t=1}^{b} \mathcal{B}_{\t}(t)| ~\overline{\mathcal{A}_{\t}}(b)}$ in terms of $f_z$, $q_z$, and the fractional degree of nodes at distance $2$ form $z$.


\begin{lemma}
For any non-stochastic node $\t$ we have:
\begin{eqnarray}
\label{eq:secondprio_lowerbound}
 \P{\cup_{t=1}^{b} \mathcal{B}_{\t}(t)\Big| \overline{\mathcal{A}_{\t}}(b)} \geq  \frac{1}{b} \sum_{t=1}^{b} \sum_{y \sim \t} \int_{x \in J_{(y,z)} \setminus I_{(y,z)}} \left(1 - \e{\frac{t f_{\fp{\s}{x}}}{b}} \right) dx \left[1 - \frac{q_{\t}}{b}(b-t)\right],
\end{eqnarray}
\end{lemma}
\begin{proof}
Using inclusion-exclusion principle, we have:
\begin{eqnarray}
\label{eq:inclusion_exclusion}
\P{\cup_{t=1}^{b}\mathcal{B}_{\t}(t)| ~\overline{\mathcal{A}_{\t}}(b)} &\geq & \sum_{t = 1}^{b} \P{\mathcal{B}_{\t}(t)| ~\overline{\mathcal{A}_{\t}}(b)} - \sum_{1 \leq t < u \leq b} \P{\mathcal{B}_{\t}(t) \cap \mathcal{B}_{\t}(u)| ~\overline{\mathcal{A}_{\t}}(b)}\nonumber\\
&=& \sum_{t = 1}^{b} \P{\mathcal{B}_{\t}(t)| ~\overline{\mathcal{A}_{\t}}(b)} \left[1- \sum_{t < u \leq b}  \P{\mathcal{B}_{\t}(u)| ~\overline{\mathcal{A}_{\t}}(b) \cap \mathcal{B}_{\t}(t)}\right]\nonumber
\end{eqnarray}
It is sufficient to upper-bound $\P{\mathcal{B}_{\t}(u)| ~\overline{\mathcal{A}_{\t}}(b) \cap \mathcal{B}_{\t}(t)}$, and to lower-bound $\P{\mathcal{B}_{\t}(t)| ~\overline{\mathcal{A}_{\t}}(b)}$. We start by showing the former, the latter is proved in Lemma \ref{lem:piz}.

The probability that $\t$ is tried at time $u$ conditioned on the event $\overline{\mathcal{A}_{\t}}(b)$ is at most the probability that a ball of type $\s$ arrives, where $\s\sim \t$, and a number $x \in J_{(y,z)} \setminus I_{(y,z)}$  is chosen.  Note that since we are conditioning on the event that
$\t$ is not tried as a first priority, the sampled point cannot belong to $J_{(y,z)} \cap I_{(y,z)}$.
By the definition of $q_z$, the total length of the intervals $J_{(y,z)} \setminus I_{(y,z)}$ for all $y~\sim z$ is $q_z$.

Conditioning on the event $\overline{\mathcal{A}_{\t}}(b)$ implies that
during the run of the algorithm, no ball arrives for the subintervals $I_{(y,z)}$ for $y\sim z$. This condition is equivalent to reducing the rate of any such nodes $\s$  by $f_{(y,z)}$.
In other words, we choose a point in subintervals with total length of $b - f_{\t}$.
Hence,
 the probability that $\t$ is tried at time $u$ conditioned on event $\overline{\mathcal{A}_{\t}}(b)$ is at most $\frac{q_{\t}}{b-f_z}$.
Since $t < u$, regardless of whether the event ${\cal B}_z(t)$ happens or not, the probability of ${\cal B}_z(u)$ cannot exceed $\frac{q_z}{b-f_z}$ (i.e. $\P{\mathcal{B}_{\t}(u)| ~\overline{\mathcal{A}_{\t}}(b) \cap \mathcal{B}_{\t}(t)}\leq \frac{q_z}{b-f_z}$). Since $f_z\leq 1$ we can approximate this by $\frac{q_z}{b}$ with an error term of $o(1/b)$ which we ignore for simplicity.

In lemma \ref{lem:piz} we lower-bound $\P{\mathcal{B}_{\t}(t)| ~\overline{\mathcal{A}_{\t}}(b)}$ (see inequality \ref{eq:Exp_pi_lower_bound}. Putting these together proves the lemma.


\end{proof}

\begin{lemma}
\label{lem:piz}
For any non-stochastic node $\t$, and any time $1\leq t\leq b$ we have:
\begin{eqnarray}
\P{\mathcal{B}_{\t}(t)| ~\overline{\mathcal{A}_{\t}}(b)}
\geq \frac{1}{b} \sum_{\s \sim \t} \int_{x \in J_{(y,z)} \setminus I_{(y,z)}} \left(1 - \e{\frac{t f_{\fp{\s}{x}}}{b}} \right) dx  \label{eq:Exp_pi_lower_bound}
\end{eqnarray}
\end{lemma}
\begin{proof}
The event $\mathcal{B}_{\t}(t)$ depends on whether the bins at distance 2 from $\t$ are full or not.
In order to incorporate the effect of the allocation of these bins on the matching of $\t$ at time $t$, we  study the evolution of the density of full bins at distance two from $\t$ as follows.
For any edge $e=(y,z)$ incident to $z$, define $F_{(y,z)}(t)$ to be those areas from $J_{(y,z)} \setminus  I_{(y,z)}$ whose corresponding first priority bin is full at time $t$. In other words, $x \in F_{(y,z)}(t)$ if $\fp{\s}{x}$ is full before time $t$. Also define $\pi_{\t}(t)$ to be the sum of the length of those intervals (i.e. $\pi_{\t}(t) = \sum_{y \sim \t} \int_{x \in F_{(y,z)}(t)} 1 dx$). First we show that $\P{\mathcal{B}_{\t}(t)| ~\overline{\mathcal{A}_{\t}}(b)}  = \frac{\E{\pi_{\t}(t)}}{b-f_z}$, then we lower-bound $\E{\pi_{\t}(t)}$.

First observe that  the bin $z$ will be tried at time $t$ as a second priority iff  a ball of type $\s\sim z$ arrives, and we choose $x \in F_{(y,z)}(t)$.
Thus the conditional probability that bin $\t$ is tried at time $t$ as the second priority is
$\P{\mathcal{B}_{\t}(t)| ~\overline{\mathcal{A}_{\t}}(b) \wedge \pi_{\t}(t)}  = \frac{\pi_{\t}(t)}{b-f_z}$.
We illustrate this through an example. In  the graph of Figure \ref{fig:shift_exmaple} let $e_1$ be the only edge adjacent to $z_1$. Suppose at time $t$, $z_2$ is full and $z_3$ is empty; we want to compute $\P{\mathcal{B}_{\t_1}(t)| ~\overline{\mathcal{A}_{\t_1}}(b) \wedge \pi_{\t_1}(t)}$. We have $F_{(y,z_1)}(t) = [0.5, 0.7]$, and $\pi_{z_1}(t) = 0.2$. Since $z_1$ will be tried as a second priority only if the arriving ball is of type $y$ and $x\in F_{(y,z_1)}$, we get $\P{\mathcal{B}_{\t_1}(t)| ~\overline{\mathcal{A}_{\t_1}}(b) \wedge \pi_{\t_1}(t)} = \frac{0.2}{b-0.5}$.
By the law of iterative expectations we obtain:
\begin{eqnarray}
\label{eq:lem:secpri}
\P{\mathcal{B}_{\t}(t)| ~\overline{\mathcal{A}_{\t}}(b)} \geq \frac{\E{\pi_{\t}(t)| \overline{\mathcal{A}_{\t}}(b)}}{b}
\end{eqnarray}

It remains to lower-bound $\E{\pi_{\t}(t)| \overline{\mathcal{A}_{\t}}(b)}$.
Using definition of $\pi_{\t}(t)$, we write $\E{\pi_{\t}(t)| \overline{\mathcal{A}_{\t}}(b)}$ as:

\begin{eqnarray}
\label{eq:pi_z}
\E{\pi_{\t}(t)| \overline{\mathcal{A}_{\t}}(b)}
 &=&   \sum_{y \sim \t} \int_{x \in J_{(y,z)} \setminus I_{(y,z)}} \E{\I{x\in F_{(y,z)}(t)}| \overline{\mathcal{A}_{\t}}(b)} dx \nonumber\\
 &=&   \sum_{y \sim \t} \int_{x \in J_{(y,z)} \setminus I_{(y,z)}} \P{x\in F_{(y,z)}(t)| \overline{\mathcal{A}_{\t}}(b)} dx
\end{eqnarray}

It suffices to lower-bound $\P{x\in F_{(y,z)}(t)| \overline{\mathcal{A}_{\t}}(b)}$.
 As explained above, $F_{(y,z)}(t)$ is a non-decreasing random process that depends on the allocation of the bins at distance 2 from $z$ at time $t$.
For $x \in J_{(y,z)} \setminus I_{(y,z)}$, let $z' = \fp{\s}{x}$.
Note that $x\in F_{(y,z)}(t)$ iff $z'$ is full at time $t$. Thus it suffices to compute the probability that $z'$ is full at time $t$.
Observe that if $z'$ is full at time $t$, it has been tried at least once as a first or second priority bin. Therefore, the probability of $z'$ being full at time $t$ is at least the probability of event ${\cal A}_{z'}(t)$. For simplicity, we ignore the possibility of the trial of $z'$ as a second priority and obtain the following lower bound:

\begin{align*}
\P{x\in F_{(y,z)}(t)| ~ \overline{\mathcal{A}_{\t}}(b)}
\geq \P{\mathcal{A}_{\fp{\s}{x}}(t)| ~ \overline{\mathcal{A}_{\t}}(b)}
\geq 1 - \e{\frac{tf_{\fp{\s}{x}}}{b}}.
\end{align*}
where the last inequality follows from \eqref{eq:priority_one}. Substituting the RHS into \eqref{eq:pi_z} and using \eqref{eq:lem:secpri} imply the Lemma.
\end{proof}

Putting equations \eqref{eq:Node_Exp}, \eqref{eq:priority_one}, \eqref{eq:secondprio_lowerbound} together and using $e^{-f_\t}\geq \e{1}$, we can lower bound the competitive ratio of Algorithm \ref{alg:adapt}:

\begin{align}
\label{eq:goalcoupled}
\frac{\E{\alg}}{\E{\opt}} \geq \frac{\sum_{\t \in \T} \left\{\left(1-\e{f_{\t}} \right) + \e{1} \left[\frac{1}{b} \sum_{t=1}^{b} \sum_{y \sim \t} \int_{x \in J_{(y,z)} \setminus I_{(y,z)}} \left(1 - \e{\frac{t f_{\fp{\s}{x}}}{b}} \right) dx \left[1 - \frac{q_{\t}}{b}(b-t)\right] \right] \right\}}{\sum_{\t \in \T} f_{\t}}
\end{align}

In the rest of the proof we show that the ratio attains its minimum when the fractional degree of all
non-stochastic nodes are exactly one, i.e., $f_{\t} = 1$, $\forall \t \in \T$. As a warm up, we first analyze this extreme case.
We have:
\begin{eqnarray}
\frac{\E{\alg}}{\E{\opt}} &\geq  & (1-e^{-1}) + e^{-1}\left[\frac{1}{b} \sum_{t=1}^b \sum_{y\sim \t} \int_{x\in J_{(y,z)}\setminus I_{(y,z)}} (1-e^{-\frac{t}{b}}) dx [1-\frac{q_\t}{b}(b-t)] \right]\nonumber\\
&=& (1-e^{-1}) + e^{-1}\left[\frac{q_z}{b} \sum_{t=1}^b (1-e^{-\frac{t}{b}})[1-\frac{q_\t}{b}(b-t)] \right]\nonumber\\
&\geq& 1-e^{-1}+q_z e^{-2} - e^{-1}q_z^2(\frac12-e^{-1}) \geq 0.702,
\label{eq:degree1-1case}
\end{eqnarray}
where the last inequality follows from the observation that for bins with $f_\t = 1$ we have $q_\t \geq \ln{2}$ (see Lemma \ref{lem:q_z} for a proof).


In the remaining parts of the proof we need to show if the fractional degree of some bins are much smaller than 1, still the competitive ratio of the algorithm remains larger than $0.702$.
Unfortunately, the dependencies between the fractional degree of $\t$
and  bins at distance 2 from $\t$ result in a significant change in the probability of $\t$ being matched as a second priority. In particular, if all of
the bins at distance 2 from $z$ have a very small rate (i.e. if $f_{\t_1}\simeq \frac{1}{n}$), then $\P{\cup_{t=1}^b {\cal B}_\t(t)| \overline{{\cal A}_\t(b)}} = O(\frac{1}{n})$. This implies that
we can not lower bound the RHS of \eqref{eq:goalcoupled} by lower bounding the worst matching probability of a bin.
Instead, in the following lemma we write the probability of $z$ being tried as a second priority bin in terms of a {\bf linear} function of $f_{\t},q_\t$ and the fractional degree of bins at distance 2 from $\t$. This will enable us to lower-bound the RHS of \eqref{eq:goalcoupled} by a node based ratio:


\begin{lemma} \label{lem:priority_two_approx}
For any non-stochastic node $\t$, 
we have:
\begin{align} \label{eq:priority_two_no_two}
\P{\cup_{t=1}^{b} \mathcal{B}_{\t}(t)| ~\overline{\mathcal{A}_{\t}}(b)}  \geq  q_{\t} \e{1} - q^2_{\t} \left(\frac{1}{2} - \e{1} \right)
- \e{1} \sum_{y \sim \t} \int_{x \in J_{(y,z)} \setminus I_{(y,z)}} \left[1 - f_{\fp{\s}{x}}\right] dx.
\end{align}
\end{lemma}
\begin{proof}
The proof of this lemma is mainly algebraic. 
First note that we can write equation \eqref{eq:secondprio_lowerbound} as
\begin{eqnarray}
\P{\cup_{t=1}^{b} \mathcal{B}_{\t}(t)| ~\overline{\mathcal{A}_{\t}}(b)}  \geq
 \sum_{y \sim \t} \int_{x \in J_{(y,z)} \setminus I_{(y,z)}} C(f_{\fp{\s}{x}}, q_{\t})  dx
\label{eq:app:priority_two_no_two}
\end{eqnarray}
where $C(f_{\fp{\s}{x}}, q_{\t}) := \frac{1}{b}\sum_{t=1}^{b}  (1 - \e{\frac{t f_{\fp{\s}{x}}}{b}} ) (1 - \frac{q_{\t}}{b}(b-t))$,  is a concave function of $f_{\fp{\s}{x}}$; this follows
from the fact that $C(.,q_\t)$ is a weighted sum of exponential functions with negative weights.  Therefore,
we can lower-bound $C(\cdot,q_\t)$ by a linear function of $f_{\fp{\s}{x}}$. Since  $0 \leq f_{\fp{\s}{x}} \leq 1$ we have:
\[
C(f_{\fp{\s}{x}}, q_{\t}) \geq C(0, q_{\t}) + [C(1, q_{\t}) - C(0, q_{\t})] f_{\fp{\s}{x}} = C(1, q_{\t})  f_{\fp{\s}{x}},
\]
where the last equality follows by the observation that $C(0, q_{\t}) = 0$. On the other hand, we have $C(1, q_{\t}) = \e{1} - q_{\t}(1/2 - \e{1})$. Therefore:
\begin{eqnarray}
C(f_{\fp{\s}{x}}, q_{\t}) \geq  (e^{-1} - q_\t(\frac12-e^{-1})) f_{\fp{\s}{x}} \geq \e{1} - q_{\t}(1/2 - \e{1}) - \e{1}[1- f_{\fp{\s}{x}}]\nonumber
\end{eqnarray}
The lemma simply follows from substituting the above equation in \eqref{eq:app:priority_two_no_two},  and using the definition of $q_\t$.
\end{proof}
Substituting \eqref{eq:priority_two_no_two} in \eqref{eq:goalcoupled}, we get:
\begin{eqnarray*}
\frac{\E{\alg}}{\E{\opt}} &\geq& \frac{\sum_{\t \in \T} \left\{(1-\e{f_{\t}}) + q_{\t} \e{2} - q_{\t}^2 \e{1} (\frac{1}{2} - \e{1}) - \e{2} \sum_{e \sim \t} \int_{x \in J_e \setminus I_e} \left[1 - f_{\fp{\s}{x}}\right] dx \right\}}{\sum_{\t \in \T} f_{\t}}
\end{eqnarray*}
Next we rearrange the last term of the numerator to eliminate all dependencies between the fractional degree of $\t$ and the bins at distance 2 from $\t$. This enables us to analyze the competitive ratio of the algorithm by the worst case ratio among all bins.
We can write:
\begin{eqnarray*}
\sum_{\t\in \T} \sum_{y\sim \t}\int_{x\in J_{(y,z)}\setminus I_{(y,z)}} [1-f_{\fp{\s}{x}}]dx =
\sum_{\t\in\T} \sum_{y\sim\t}\int_{x\in I_{(y,z)}\setminus J_{(y,z)}} [1-f_{\fp{\s}{x}}]dx,
\end{eqnarray*}
Here the equality follows from the observation that for all $y\in Y$, both sides are integrating
over all $x\in [0,r_y]$ where $\fp{y}{x}\neq\sp{y}{x}$.
Since for any $x\in I_{(y,z)}\setminus J_{(y,z)}$, we have $\fp{\s}{x} = z$,  and
\begin{eqnarray*}
\sum_{\t\in\T} \sum_{y\sim\t}\int_{x\in I_{(y,z)}\setminus J_{(y,z)}} [1-f_{\fp{\s}{x}}]dx
= \sum_{\t\in\T} \sum_{y\sim\t}\int_{x\in I_{(y,z)}\setminus J_{(y,z)}} [1-f_\t]dx\leq \sum_{\t\in\T} f_{\t}[1-f_\t].
\end{eqnarray*}

Therefore,  the competitive ratio of the algorithm is at least:
\begin{eqnarray}
\frac{\E{\alg}}{\E{\opt}}
&\geq & \min_{\t \in \T} \frac{(1-\e{f_{\t}}) + q_{\t} \e{2} - q_{\t}^2 \e{1} (\frac{1}{2} - \e{1})-\e{2} f_{\t}\left[1 - f_{\t}\right] }{f_{\t}}.
 \label{eq:adpt_ratio_no_two}
\end{eqnarray}

Since for $0\leq q_\t\leq 1$, the RHS is an increasing function of $q_\t$, any lower-bound on $q_z$ also gives a lower-bound on the competitive ratio of the algorithm. In particular, if $f_\t\leq \frac12$, we can lower-bound $q_\t$ by zero and we get $\frac{\E{\alg}}{\E{\opt}} \geq \frac{1-e^{-f_\t}-e^{-2}f_\t[1-f_\t]}{f_\t}\geq 0.719$.
On the other hand, if $f_\t\geq \frac12$ we use the lower-bound $q_\t\geq \ln{2}+f_\t-1$ (see Lemma \ref{lem:q_z} for the proof), and we obtain that the worst lower-bound is attained for bins with fractional degree 1:
$$\frac{\E{\alg}}{\E{\opt}}\geq 1-e^{-1} +e^{-2}\ln{2}-e^{-1}(\ln{2})^2(\frac12-e^{-1})  \geq 0.702.$$

This completes the proof of Theorem \ref{thm:adap}.
\end{proof}
\begin{remark}
As we discussed earlier (equation \eqref{eq:degree1-1case}) the worst  competitive ratio of the algorithm is attained for bins with fractional degree 1, thus the linear bounds used in the proof of Lemma \ref{lem:priority_two_approx} does not change worst case analysis of the algorithm.
\end{remark}

\begin{lemma} \label{lem:q_z}
For any non-stochastic node $\t$, we have
$q_{\t} \geq \ln 2 + f_{\t} -1$
\end{lemma}
\begin{proof}
The proof  follows from Observation \ref{obs:priority2_slots} and an optimization over the
sampling rate of the neighboring stochastic nodes. 

\begin{figure}[htb]
\begin{centering}
%
%
%
%

\def\radius{2.6}
\def \Pointsize {1.4pt}
\begin{tikzpicture}[pre/.style={<-,shorten <=1.5pt,>=stealth,thick}, post/.style={->,shorten >=1pt,>=stealth,thick}]
\tikzstyle{every node}=[draw,shape=rectangle,minimum size=5mm, inner sep=0];
\path (0,0) node (z) {$z$};
\tikzstyle{every node}=[shape=circle,minimum size=6mm, inner sep=0];
\path (-\radius,+\radius) node [draw] (y1) {$y_1$} ++ (0,.5) node  {$r_{y_1}=.5$};
\path (0,+\radius) node [draw] (y2) {$y_2$} ++ (0,.5) node {$r_{y_2}=.5$};
\path (\radius,\radius) node [draw] (y3) {$y_3$} ++ (0,.5) node {$r_{y_3}=1$};

\tikzstyle{every node}=[];
\path[-] (z) edge node [left] {$f_{(y_1,z)}=.3$} (y1)
(z) edge node  {$f_{(y_2,z)}=.3$} (y2)
(z) edge node [right] {$f_{(y_3,z)}=.4$} (y3);
\path (6, 2.5) node {$E^o_z = \{ (y_1,z), (y_2,z)\}$};
\path (6, 1) node {$E^n_z = \{ (y_3,z)\}$};

\end{tikzpicture}
\end{centering}
\caption{An example of a non-stochastic node $\t$ with $|E^o_\t|>1$.}
\label{fig:rates}
\end{figure}
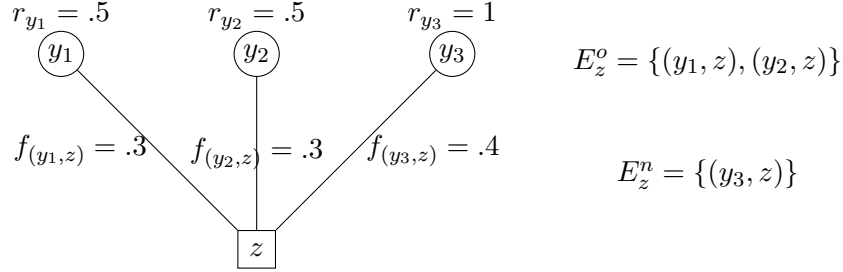

Let $E_\t$ be the set of edges incident to $\t$ in graph $G$.
We partition $E_{\t}$ into two subsets $E^{o}_\t$ and $E^{n}_{\t}$, such that $E^{o}_{\t}$ consists of edges $(\s,\t)$ where $f_{(\s,\t)}>\frac{1}{2} r_{\s}$, and $E^n_{\t}=E_\t\setminus E^o_\t$ are the rest of the edges. In words,  $E^{o}_{\t}$ is the set of edges $e$ for which $I_{(y,z)}$ and $J_{(y,z)}$ {\em overlap}.
For example, if the rates of all stochastic nodes are $1$, for any edge $(y,z)\in E^o_\t$ we must have $f_{(y,z)}>\frac{1}{2}$; but since $f_\t\leq 1$ we must have $|E^o_\t|\leq 1$. However, this is not necessarily true if we allow the stochastic
nodes to have arbitrary rates (see Figure \ref{fig:rates} for an example).
By Observation \ref{obs:priority2_slots}, we have:
\begin{equation}
\label{eq:qfnrofo}
q_{\t} = \sum_{y \sim \t} \int_{x \in J_{(y,z)} \setminus I_{(y,z)}} 1 dx = \sum_{\s:(\s,\t) \in E^{n}_{\t}}f_{(\s,\t)} + \sum_{\s:(\s,\t) \in E^{o}_{\t}}r_\s - \sum_{\s:(\s,\t) \in E^{o}_{\t}} f_{(\s,\t)}.
\end{equation}
Let $f^n_\t$, $r^o_\t$, and $f^o_\t$ be the first, second, and the third summations in the RHS, i.e.,
$q_{\t} = f^{n}_\t + r^o_\t - f^o_\t$. By Proposition \ref{prop:OPT} we can show $f^o_\t \leq (1- \e{r^o_\t})$;
it is sufficient to replace all stochastic neighbors of $\t$ with a super node $\s^*$ of rate $r^o_\t$,
and use Proposition \ref{prop:OPT} to conclude that
\begin{equation}
\label{eq:qsupernode}
f^o_\t = f_{(\s^*,\t)}\leq (1-e^{-r_{\s^*}})=(1-e^{-r^o_\t}).
\end{equation}
We can obtain a lower bound on $q_{\t}$ simply by using the above equations and noting that $f_\t=f^n_\t+f^o_\t$:
\begin{eqnarray*}
 q_z &=& f^n_\t+r^o_\t-f^o_\t = f_\t+r^o_\t-2f^o_\t\geq f_\t+r^o_\t-2(1-e^{-r^o_\t})\\
&\geq& f_\t+\ln{2}-2+2e^{-\ln{2}}=f_\t+\ln{2}-1,
\end{eqnarray*}
where the first equality follows from equation \eqref{eq:qfnrofo}, the first inequality follows from equation \eqref{eq:qsupernode}, and the second inequality follows from the fact that $r^o_\t=\ln{2}$ is the minimizer of
$r^o_\t+2e^{-r^o_\t}$.
\end{proof}

\begin{corollary}
If we restrict the sampling rates of all stochastic nodes to be integral (i.e. $r_\s$'s are integral),
then the competitive ratio of Algorithm \ref{alg:adapt} is at least $0.705$.
\end{corollary}
\begin{proof}
The corollary simply follows from a better lower-bound on $q_\t$ in terms of $f_\t$. Since the rates are integral we can show $q_\t\geq f_\t+2e^{-1}-1$; in particular, in the proof of Lemma \ref{lem:q_z} assuming integral rates, we get $r^o_\t\in \{0,1\}$ which  implies that $q_\t\geq f_\t+1-2(1-e^{-1})=f_\t+2e^{-1}-1$. Then the corollary follows from plugging this lower-bound into equation \eqref{eq:adpt_ratio_no_two}.
\end{proof}


\section{Upper Bounds for Online Algorithms}
\label{sec:hardness}

We will present three examples. The first example gives a straightforward $1-1/e$ upper bound for the performance of {\em non-adaptive randomized} algorithms. It shows that when the rates are arbitrarily small, no non-adaptive algorithm can achieve a competitive ratio better than $1-1/e$.
Note that a randomized non-adaptive algorithm predetermines {\em distribution} ${\cal{D}}_{\s,i}$ for the $i$-th arrival of type $\s$. In
other words, when the $i$-th ball of type $\s$ arrives it will be matched to the neighbor bin $\t$ with probability $\PP{{\s,i}}{\t}$.

\begin{proposition}
\label{prop:hardness_fractionalrate}
There is an instance of the online stochastic matching problem with small rates, $r_{\s} = o(1)$, for which no non-adaptive randomized algorithm can achieve a competitive ratio better than $1-\e{1}$.
\end{proposition}

\begin{proof}
Suppose $G(\S,\T,E)$ is a complete bipartite graph, where $|\S|=n^2$ and $|\T|=n = b$; also suppose the rate
of all types is $1/n$.
Since $G$ is a complete bipartite graph, $OPT$ can easily allocate all the arriving balls and $\eopt=n$.
On the other hand, since $r_{\s} = o(1)$, with high probability, there will be at most one ball of each type.
Therefore, any non-adaptive randomized algorithm only needs to predetermine one distribution ${\cal D}_{\s,1}$
for each type $\s$. For each bin $\t\in \T$, let $p_{\t}$ be the probability that an incoming ball is matched to $\t$.
In other words,
$$ p_{\t} = \sum_{\s\in\S} r(\s)\PP{{\s,1}}{\t} = \frac{1}{n}\sum_{\s\in\S}\PP{{\s,1}}{\t}$$
With probability
$\e{p_{\t}}$ no ball will be matched to the bin $\t$ in the run of the process.
Thus, $\ealg = \sum_{\t\in \T} (1-\e{p_{\t}})$.
Since function $1-\e{x}$ is concave we have:
\begin{align*}
\frac{\ealg}{\eopt} =  \frac{\sum_{\t\in \T} (1-\e{p_{\t}})}{n}
\leq   (1-\e{\frac{1}{n}\sum_{\t\in \T} p_{\t}}).
\end{align*}
On the other hand, we have:
\begin{align*}
\frac{1}{n}\sum_{\t\in \T} p_{\t} = &~\frac{1}{n^2}\sum_{\t\in\T} \sum_{\s\in \S} \PP{{\s,1}}{\t}
= \frac{1}{n^2} \sum_{\s\in \S} \sum_{\t\in\T} \PP{{\s,1}}{\t} =1.
\end{align*}
Therefore, $\frac{\ealg}{\eopt} \leq 1 - \e{1}$ which
completes the proof.
\end{proof}

Our next two examples give an upper bound on the performance of {\em any} deterministic or randomized online algorithm. In the first example, the rates are integral. Our upper bound of $1-e^{-2}$ is slightly better than the result of \cite{bahmani}.

\begin{proposition}
\label{prop:hardness_adapt}
There exists an instance of the online stochastic matching problem with integral rates for which no online algorithm can achieve an expected competitive ratio better than $1-e^{-2}\simeq 0.86$.
\end{proposition}
\begin{proof}
Construct a bipartite graph $G(\S,\T,E)$, where $\S=\S_1 \cup \S_2$, $|\S_1| = |\T| = n$, and $|\S_2| = n/e$.
The set $E$ of edges consists of a perfect
matching between the vertices of  $\S_1$ and $\T$ denoted by $E_1$, and a complete bipartite
graph between $\S_2$ and $\T$, denoted by $E_2$. See Figure \ref{fig:counter_exmaple}.

First, we prove that $\E{\opt} = n$. Given the
sequence of arrivals, first we match through the perfect matching ($E_1$). In other
words, we match one ball of each type $\s_1 \in \S_1$. Note that with probability
$\e{1}$, there will be no ball of type $\s_1$, thus, in expectation, $(1-1/e)$ fraction
of the bins will remain empty after matching through $E_1$. On the other hand, the
expected number of balls of types $\S_2$ is $n/e$, which can be matched
with the $n/e$ empty bins through the edges of the complete bipartite graph, $E_2$.
Hence, this simple scheme finds the maximum matching and $\E{\opt} = n$.

\begin{figure*}
\centering
    \psfrag{ry1}{\textcolor{black}{\scriptsize{$r_{\s_1} = 1$}}}
    \psfrag{Y1}{\textcolor{black}{\scriptsize{$|\S_1| = n$, $f_{\s_1} = 1-1/e$, $f_{e_1} = 1-1/e$}}}
    \psfrag{ry2}{\textcolor{red}{\scriptsize{$r_{\s_2} = 1$}}}
    \psfrag{Y2}{\textcolor{red}{\scriptsize{$|\S_2| = n/e$, $f_{\s_2} = 1$, $f_{e_2} = 1/n$}}}
    \psfrag{Z}{\textcolor{black}{\scriptsize{$|\T| = n$, $f_{\t} = 1 $}}}
    \psfrag{phi}{\textcolor{black}{\scriptsize{$\Phi(t)$}}}
    \psfrag{psi}{\textcolor{black}{\scriptsize{$\Psi(t)$}}}
	\includegraphics[height=.25\textheight,clip]{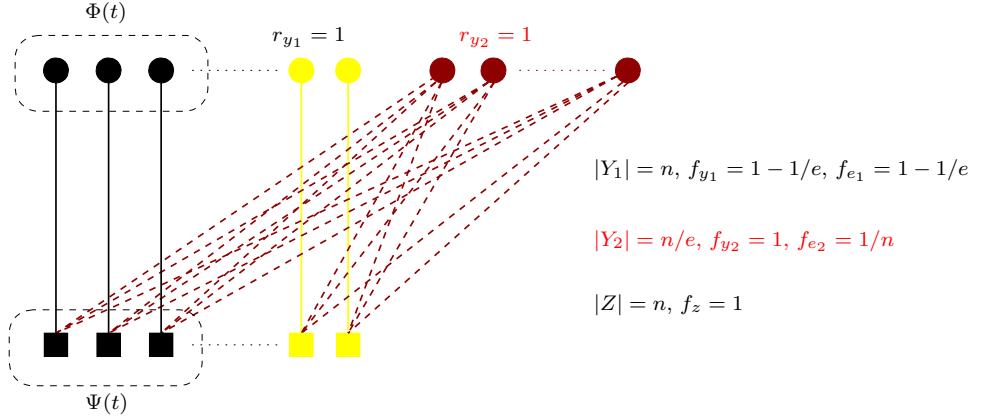}
\caption{Graph structure for the proof of Proposition \ref{prop:hardness_adapt}}
\label{fig:counter_exmaple}
\end{figure*}

On the other hand, consider an arbitrary online algorithm $\alg$; at time $t$, let $\Psi(t) \subseteq \T$ be the set of
full (matched) bins, and $\Phi(t) \subseteq \S_1$ be the set of types
that have a neighbor in $\Psi(t)$. If the $(t+1)$-st ball is of type $\Phi(t)$, it is impossible for $\alg$ to match this ball. Thus:
\begin{align}
|\Psi(t+1)| \leq & |\Psi(t)|
\label{eq:psi_recursion}
+ \I{\tm{$t+1^{st}$ ball is not of type $\Phi(t)$}}
\end{align}
Observe that,
\begin{align} \label{eq:prob_Phi}
\P{\tm{$t+1^{st}$ ball is not of type $\Psi(t)$}} = 1-\frac{ |\Psi(t)|}{n(1+1/e)}.
\end{align}
Note that $b = n(1+1/e)$ balls will arrive, thus $\E{\alg} = \E{|\Psi(n(1+1/e))|}$. Taking expectations from both sides
of \eqref{eq:psi_recursion} and using \eqref{eq:prob_Phi} result in:
\[
\E{\alg} \leq n(1+1/e) \times (1-1/e) = (1-\frac{1}{e^2}) n,
\]
which proves the claim of the proposition.
\end{proof}


Our last and probably most interesting example is for general online algorithms, under arbitrary rates. In this example, we use calculations on the size of perfect matchings in random bipartite graphs studied earlier in the context of Random SAT and cuckoo hashing \cite{dgmmpr09,fm09,fk09}.

For a set $Z$ of bins, define $Y_k$ to be a set of ${|Z|\choose k}$ vertices, each connected to a distinct subset of cardinality $k$ of $Z$. These sets will play an important role in constructing examples with large competitive ratio. Let us start with a simple example. Consider an instance of online stochastic matching where $\S=\S_3$, $|\T|=n$. Also suppose that all the rates are equal  and  $b=0.9n$, i.e.  the rate of each ball $r_\s=n/{0.9n\choose 3}$.

From the perspective of the algorithm, we will have a sequence of $0.9n$ arriving balls each connected to three bins chosen independently and uniformly at random. Because of that, all the empty bins are  equivalent; thus the online algorithm can assign the arriving ball to any of its unoccupied neighbors, if there is any.  Similar to the proof of Proposition \ref{prop:hardness_adapt}, let $\Psi(t)\subseteq \T$ be the set of full bins at time $t$,  and $\Phi(t)\subseteq Y$ be the set of types of balls that have no neighbor in $\T\setminus \Psi(t)$ at time $t$. Note that if the $(t+1)$-st ball is of type $\Phi(t)$, it is impossible for any online algorithm to match it. Note that:
\begin{eqnarray*}
\label{eq:cuckoo_t3}
\P{\tm{$t+1^{st}$ ball is not of type $\Phi(t)$}} = 1-\frac{|\Phi(t)|}{{n\choose 3}}
 = 1-\frac{{{|\Psi(t)|}\choose 3}}{{n\choose 3}}.
\end{eqnarray*}
Thus we can simply write a recurrence relation to compute the expected performance of the online algorithm.

The more difficult part is to compute the optimum solution. The optimum offline algorithm will essentially find the maximum matching between all arrived ball types and the bins. The size of this maximum matching is studied by Path and Rodler \cite{pr04}. There, the problem is defined as follows: there are $b$ keys to be hashed into $n$ buckets, each capable of holding a single key. Each key has $k\geq 2$ (distinct) associated buckets chosen uniformly at random and independently of the
choices of other keys. A hash table can be constructed successfully if each key can be placed into one
of its buckets.

Define  $c^*_k$ to be the threshold such that if $b/n < c^*_k$ and $n$ is large enough, the resulting bipartite graph has a matching of size $b$.
There has been extensive effort to compute $c^*_k$ \cite{fm09, fk09, dgmmpr09}. In particular, it has been shown that $c^*_3 > 0.91$. Therefore, we can argue that if $b/|\T|=0.9 <c^*_3$ then the optimum can match all of the balls with high probability. Dietzfelbinger et al. \cite{dgmmpr09} considered an irregular version of the cuckoo hashing, where the number of choices corresponding to a key is a random variable depending on the key. In particular, they considered the case where a key has 2 choices with probability 1/2 and 3 choices with probability 1/2 (say $2.5$ choices in average), and they defined the number $c^*_{2.5}$ similarly. Interestingly, they show that $c^*_{2.5}\simeq 0.81034$ which is much larger than $c^*_2$.

In the next proposition we use a combination of the irregular cuckoo hashing idea and the idea of the proof of Proposition \ref{prop:hardness_adapt} (adding the type $\S_n$) to obtain a better upper bound on the performance of optimal online algorithms.

\begin{proposition}
\label{prop:best_example}
There is an instance of the online stochastic matching problem for which no algorithm can achieve a competitive ratio better than $0.823$.
\end{proposition}
\begin{proof}
Let $\S=\S_2 \cup \S_3 \cup \S_n$, $|\T|=n$; note that $\S_n$ and $\T$ form a complete bipartite graph. Suppose in expectation we throw $m:=1/2c^*_{2.5}n$ balls of types in $\S_2$, $m$ of types in $\S_3$ and $n-2m$ of type in $\S_n$. Therefore, we have $b=n$, and $r_\s= m/{n\choose 2}$  for $\s \in \S_2$,  $r_\s= m/{n\choose 3}$ for $\s \in \S_3$, and $r_\s = n-2m$ for $\s \in \S_n$.
The optimum offline solution would first match the balls of types in $\S_2$ and $\S_3$, and because the expected number of these balls is at most $c^*_{2.5}n$,  it can match all of them with high probability. Then, it matches all the balls of type $\S_n$ to the unoccupied bins. Therefore $\E{\opt}=n$. Let \alg ~be an online algorithm and let $\Psi(t)$ and $\Phi(t)$ be defined as above.
 Similar to the equation \eqref{eq:cuckoo_t3} we can compute the probability that an incoming ball can be matched by $\alg$. Note that if a ball of types in $\S_n$ arrives the online algorithm can always match it through the complete graph; on the other hand, if a ball of type $\S_2$ or $\S_3$ arrives it can only be matched if it has at least one neighbor in $\T\setminus \Psi(t)$. Note that:
 \begin{eqnarray*}
\P{\tm{the type of $t+1^{st}$ ball is not   in $\Phi(t)$}}
=1-\frac{m}{n} \left[\frac{{{|\Psi(t)|}\choose 2}}{{n\choose 2} } +\frac{{{|\Psi(t)}|\choose 3}}{{n\choose 3}}\right]
\end{eqnarray*}
Therefore, we have
\begin{eqnarray*}
\E{|\Psi(t+1)|} &\leq & ~\E{|\Psi(t)|} + 1
-\frac{m}{n} \E{\frac{{{|\Psi(t)|}\choose 2}}{{n\choose 2}} +\frac{{{|\Psi(t)}|\choose 3}}{{n\choose 3}}}\\
&\leq& \E{|\Psi(t)|} + 1
-\frac{m}{n} \left[\frac{{{\E{|\Psi(t)|}}\choose 2}}{{n\choose 2}} +\frac{{\E{{|\Psi(t)}|}\choose 3}}{{n\choose 3}}\right],
\end{eqnarray*}
where the last inequality follows from Jensen's inequality. One can numerically compute $\E{|\Psi(n)|}$ and show that $\E{|\Psi(n)|} \leq 0.823n$ for $n>1000$. Thus for $n>1000$, we have:
$$
\E{\alg} \leq \E{|\Psi(n)|} \leq 0.823n,
$$
which implies that the approximation ratio of the online algorithm is at most 0.823.
\end{proof}

\section{Discussion}
We should also point out that competitive analysis is not the only possible or necessarily the most suitable approach for this problem. Because the distribution from which the input is generated is known, one can use dynamic programming (or enumeration of future events) to derive the optimal allocation policy. Unfortunately, the dynamic programming approach takes exponential time. In fact, one can show that the problem of computing the optimal allocation policy in NP-hard.  We leave it as an open problem whether it is possible to come up with a polynomial-time algorithm with an approximation guarantee that is better than the best possible competitive ratio for this problem or the competitive ratio that we obtain here.


\bibliographystyle{abbrv}
\bibliography{references}

\end{document}